\documentclass[sigconf]{aamas} 

\usepackage{balance} 

\setcopyright{ifaamas}
\acmConference[AAMAS '23]{Proc.\@ of the 22nd International Conference
on Autonomous Agents and Multiagent Systems (AAMAS 2023)}{May 29 -- June 2, 2023}
{London, United Kingdom}{A.~Ricci, W.~Yeoh, N.~Agmon, B.~An (eds.)}
\copyrightyear{2023}
\acmYear{2023}
\acmDOI{}
\acmPrice{}
\acmISBN{}

\title[Title]{Random Majority Opinion Diffusion: Stabilization Time, Absorbing States, and Influential Nodes}

\author{Ahad N. Zehmakan}
\affiliation{
  \institution{The Australian National University}
  \city{Canberra}
  \country{Australia}}
\email{ahadn.zehmakan@anu.edu.au}

\begin{abstract}
Consider a graph $G$ with $n$ nodes and $m$ edges, which represents a social network, and assume that initially each node is blue or white (indicating its opinion on a certain topic). In each round, all nodes simultaneously update their color to the most frequent color in their neighborhood. This is called the
Majority Model (MM) if a node keeps its color in case of a tie and the Random Majority Model (RMM) if it chooses blue with probability $1/2$ and white otherwise.

We prove that there are graphs for which RMM needs exponentially many rounds to reach a stable configuration in expectation, and such a configuration can have exponentially many states (i.e., colorings). This is in contrast to MM, which is known to always reach a stable configuration with one or two states in $\mathcal{O}(m)$ rounds. For the special case of a cycle graph $C_n$, we prove the stronger and tight bounds of $\lceil n/2\rceil-1$ and $\mathcal{O}(n^2)$ in MM and RMM, respectively. Furthermore, we show that the number of stable colorings in MM on $C_n$ is equal to $\Theta\left(\Phi^n\right)$, where $\Phi = (1+\sqrt{5})/2$ is the golden ratio, while it is equal to 2 for RMM. Our results demonstrate how minor local alterations, such as tie-breaking rule, can significantly influence the global behavior of the process.

We also study the minimum size of a winning set, which is a set of nodes whose agreement on a color in the initial coloring enforces the process to end in a coloring where all nodes share that color. We present tight bounds on the minimum size of a winning set for both MM and RMM.

Furthermore, we analyze our models for a random initial coloring, where each node is colored blue independently with some probability $p$ and white otherwise. Using some martingale analysis and counting arguments, we prove that the expected final number of blue nodes is respectively equal to $(2p^2-p^3)n/(1-p+p^2)$ and $pn$ in MM and RMM on a cycle graph $C_n$.

Finally, we conduct some experiments which complement our theoretical findings and also lead to the proposal of some intriguing open problems and conjectures to be tackled in the future work.
\end{abstract}


\newcommand{\BibTeX}{\rm B\kern-.05em{\sc i\kern-.025em b}\kern-.08em\TeX}

\usepackage{amsmath}
\usepackage{enumerate}
\begin{document}


\pagestyle{fancy}
\fancyhead{}


\maketitle 
\section{Introduction}

When facing a decision or forming an opinion about a subject such as the quality of a technological innovation or the success of a political party, we are usually influenced by the opinion of our friends, family, colleagues, and the figures whose opinions we value. Hence, our opinions are constantly influenced and shaped through interactions with our connections. Furthermore, due to the extensive rise in the usage of online social platforms such as Facebook, Instagram, WeChat, TikTok, and Twitter, opinions are exchanged and formed at a higher pace.

Companies, political parties, and even governments attempt to leverage the power of opinion formation and influence propagation through online social platforms to reach their commercial and political goals. For example, marketing campaigns routinely use online social networks to sway people’s opinions in their favor, by targeting subsets of members with free samples of their products or misleading information. Therefore, opinion diffusion and (mis)-information spreading can affect different aspects of our lives from economics and politics to fashion and music.

There has been a fast-growing demand for a better and deeper understanding of opinion formation and information spreading processes in social networks.
A more profound knowledge of the collective decision-making and opinion diffusion processes would let us control and regulate the effect of marketing and political campaigns and stop the spread of misinformation.

The evolution of social dynamics has been a topic of intense study by researchers from a wide range of backgrounds such as economics~\cite{bharathi2007competitive}, epidemiology~\cite{pastor2001epidemic}, social psychology~\cite{yin2019agent}, and statistical physics~\cite{galam2008sociophysics}. It particularly has gained significant popularity in theoretical computer science, especially in the rapidly growing literature focusing on the interface between social choice and social networks, cf.~\cite{bredereck2017manipulating, auletta2018reasoning, huang2013impact}.

Numerous models have been proposed to simulate the opinion formation processes. It is inherently difficult to develop models which reflect reality perfectly since these processes are way too complex to be expressed in purely mathematical terms. Therefore, a suitable model strives to capture the fundamental properties of opinion spreading processes, but at the same time be simple enough to permit accurate and profound mathematical analysis. Therefore, the objective is to establish models which justifiably approximate the real opinion diffusion processes by disregarding less essential, but distracting, parameters. The analysis of such approximate models would allow researchers to shed some light on the fundamental principles and recurring patterns in the opinion diffusion processes, which are otherwise concealed by the intricacy of the full process.

Each opinion diffusion model has three essential components. Firstly, one needs to define how the interactions between the individuals take place. A well-received choice is to use a graph structure, where a node represents an individual and an edge between two nodes corresponds to a relation between the respective individuals, e.g. friendship or common interests. Secondly, there exist different options for modeling the opinion of the individuals. A popular choice is to assign a binary value, say blue or white, to each node, which indicates whether the node is positive or negative about a certain topic. Last but not the least, a crucial component of any model is its updating rule which defines how and in what order the nodes update their opinion. In the plethora of various updating rules, the majority rule, where a node chooses the most frequent opinion (i.e., color) in its neighborhood, has attracted a substantial amount of attention.


Different aspects of opinion diffusion models have been investigated, both theoretically (by exploiting the rich tool kit from graph and probability theory) and experimentally (by conducting a vast spectrum of experiments on graph data from real-world social networks). An enormous part of the research performed in this area falls under the umbrella of the following three fundamental questions:

\begin{enumerate}
    \item How long does it take for the process to reach a stable configuration, and how does such a stable configuration look?
    \item What is the minimum number of nodes which need to be blue to ensure that the whole graph eventually becomes~blue?
    \item What is the expected final number of blue nodes starting from a random initial coloring?
\end{enumerate}

In the present paper, we contribute to the study of the aforementioned questions for two of the most basic majority based models on general graphs and special classes of graphs, in particular cycles.

\textbf{Roadmap.} In the rest of this section, we first provide some basic definitions which create the ground to describe our contributions in more depth; then, we give a brief overview of the relevant prior work. Our theoretical findings to address questions (1), (2), (3) are presented in Sections~\ref{stabilizarion},~\ref{winning set},~\ref{random initial}, respectively. Finally, our experimental results are provided in Section~\ref{experiments}.

\subsection{Preliminaries}
\label{preliminaries}
\textbf{Graph Definitions.} Let $G=\left(V,E\right)$ be a simple connected undirected graph and define $n:= |V|$ and $m :=|E|$. For a node $v\in V$, $N\left(v\right):=\{u\in V: \{u,v\} \in E\}$ is the \emph{neighborhood} of $v$. For a set $S\subset V$, we define $N_S\left(v\right):=N\left(v\right)\cap S$. Moreover, $d\left(v\right):=|N\left(v\right)|$ is the \emph{degree} of $v$ and $d_S\left(v\right):=|N_S\left(v\right)|$. Note that whenever graph $G$ is not clear from the context, we add a superscript, e.g. $d^G(v)$.

\textbf{Models.} For a graph $G$, a \emph{coloring} is a function $\mathcal{C}:V\rightarrow\{b,w\}$, where $b$ and $w$ represent blue and white. For a node $v\in V$, the set $N_a^{\mathcal{C}}\left(v\right):=\{u\in N\left(v\right):\mathcal{C}\left(u\right)=a\}$ includes the neighbors of $v$ which have color $a\in\{b,w\}$ in the coloring $\mathcal{C}$. Furthermore, we write $\mathcal{C}|_S=a$ for a set $S\subseteq V$ if $\mathcal{C}(v)=a$ for every $v\in S$.

Assume that we are given an initial coloring $\mathcal{C}_0$ on a graph $G$. In a model $M$, $\mathcal{C}_t\left(v\right)$, which is the color of node $v$ in round $t\in\mathbb{N}$, is determined based on a predefined updating rule. We are interested in the \textit{Majority Model (MM)} where the updating rule is as follows:

$\mathcal{C}_t(v)$ = 
$\begin{cases} 
\mathcal{C}_{t-1}(v) \quad if \ 
|N_b^{\mathcal{C}_t-1}(v)| = |N_w^{\mathcal{C}_{t-1}}(v)|\\
argmax_{a \in \{b, w\}}|N_a^{\mathcal{C}_{t-1}}(v)| \quad otherwise\\
\end{cases}.$

In other words, each node chooses the most frequent color in its neighborhood and keeps its color in case of a tie. The \textit{Random Majority Model (RMM)} is the same as MM except that in case of a tie, a node chooses one of the two colors independently and uniformly at random.

In these models, we define $b_t$ and $w_t$ for $t\in \mathbb{N}_0$ to be the number of blue and white nodes in $\mathcal{C}_t$. These correspond to random variables in RMM and also in MM when the initial coloring is random.

We say the process reaches the \emph{blue (white) coloring} if it reaches the coloring where all nodes are blue (white). For a cycle graph $C_n$ with even $n$, there are two colorings where every two adjacent nodes have different colors. We call these two colorings the \emph{alternating colorings}. If MM or RMM process reaches one of the two alternating colorings, it keeps switching between them. We say the process has reached the \emph{blinking configuration}.

For a graph $G$, we say that a coloring $\mathcal{C}$ is \emph{stable} if \textit{one} application of MM (similarly RMM) on $\mathcal{C}$ deterministically outputs $\mathcal{C}$. (For RMM, this implies that there are no ties.) Note that a stable coloring need not be monochromatic. Furthermore, a \emph{$p$-random coloring,} for $0\le p \le 1$, is a coloring where each node is colored blue independently with probability (w.p.) $p$
 and white otherwise.

\textbf{Stabilization Time and Periodicity.} Since the updating rule in MM is deterministic and there are $2^n$ possible colorings, for any initial coloring the process reaches a cycle of colorings and remains there forever. The number of rounds the process needs to reach the cycle is the \textit{stabilization time} and the length of the cycle is the \textit{periodicity} of the process.

RMM on an $n$-node graph $G$ corresponds to a Markov chain. This Markov chain has $2^n$ states (i.e., $2^n$ possible colorings) and there is an edge from state $s$ to $s'$ if there is a non-zero probability to go from $s$ to $s'$ in RMM. Since this is a directed graph, its state set can be partitioned into maximal strongly connected components. (A state set is a strongly connected component if every state is reachable from every other state, and it is maximal if the property does not hold when we add any other state to the set.) Furthermore, we say a maximal strongly connected component is an absorbing component if it has no outgoing edge. If each maximal strongly connected component is contracted to a single state, the resulting graph is a directed acyclic graph. This implies that in RMM, regardless of the initial coloring, the process eventually reaches an absorbing component and remains there forever. The expected number of rounds the process needs to reach an absorbing component is the \textit{stabilization time} and the size of the absorbing component is the \textit{periodicity} of the process. In simple words, the process eventually reaches a subset of states (colorings) and keeps transitioning between them. The stabilization time is the expected number of rounds to get there, and the periodicity is their number.

\textbf{Winning and Resilient Sets.} For MM or RMM on a graph $G=(V,E)$, we say a node set $S \subseteq V$ is a \textit{winning set} whenever the following holds: If all nodes in $S$ are blue (white), then the process eventually reaches the blue (white) coloring regardless of the color of nodes in $V\setminus S$ and all the random choices (in RMM). Furthermore, we say a node set $S \subseteq V$ is a \textit{resilient set} whenever the following holds: If $S$ is fully blue (white) then all nodes in $S$ remain blue (white) forever, regardless of the color of the other nodes and the random choices. We observe that a set $S$ is resilient in MM (resp. RMM) if and only if for every node $v\in S$, $|N_{S}(v)|\ge d(v)/2$ (resp. $|N_{S}(v)|> d(v)/2$).

\textbf{Path Partition.} Consider a cycle $C_n$ and a coloring $\mathcal{C}$. We say a path is blue (white) if all its nodes are blue (white). A path is \emph{monochromatic} if it is blue or white. Furthermore, a path is \emph{alternating} if every two adjacent nodes have opposite colors. The length of a path is its number of nodes and an even (odd) path is a path whose length is even (odd). Except when $n$ is even and $\mathcal{C}$ is one of the two alternating colorings, there must exist at least one monochromatic path of length two or larger. Let $B$ (resp. $W$) be the set of nodes on the maximal blue (resp. white) paths of length at least two in $\mathcal{C}$. Then, all the nodes which are not in $B\cup W$ can be partitioned into maximal alternating paths, which are surrounded by the aforementioned monochromatic paths. We call the union of these maximal monochromatic and alternating paths, the \emph{path partition} in $\mathcal{C}$.


\textbf{McDiarmid's Inequality.} We use an extension of McDiarmid's inequality which gives a bound on the input sensitivity of random variables when differences in the output satisfy some bound.

\begin{definition}
Let $X:\Omega\rightarrow \mathbb{R}$ be a random variable over the probability space $\Omega=\{0,1\}^n$. We say $X$ is \emph{difference-bounded} by $(\beta,c,\delta)$ if the following holds: (i) there is a ``bad'' subset $B\subset \Omega$, where $|B|/|\Omega|=\delta$ (ii) if $\omega,\omega'\in \Omega$ differ only in the $i$-th coordinate, and $\omega\notin B$, then $|X(\omega)-X(\omega')|\le c$ (iii) for any $\omega$ and $\omega'$ differing only in the $i$-th coordinate, $|X(\omega)-X(\omega')|\le \beta$.
\end{definition}

\begin{theorem}[An Extension of McDiarmid's Inequality~\cite{kutin2002extensions}]
\label{McDiarmid-thm}
Let random variable $X:\{0,1\}^n\rightarrow \mathbb{R}$ be difference-bounded by $(\beta,c,\delta)$, then for any $\epsilon>0$, the probability $\mathbb{P}[(1-\epsilon)\mathbb{E}[X]\le X\le (1+\epsilon)\mathbb{E}[X]]$ is at least $1-2\exp\left(\frac{-\epsilon^2\mathbb{E}[X]^2}{8nc^2}\right)-\frac{2\delta n\beta}{c}$. 
\end{theorem}

\textbf{With High Probability.} We assume that $n$ (i.e., $|V|$) tends to infinity. We say an event happens with high probability (w.h.p.) when it occurs w.p. $1-o(1)$.
\subsection{Our Contribution}
\textbf{Contribution 1: Stabilization Time and Periodicity.} It is known \cite{poljak1986pre} that the stabilization time in MM on a graph $G$ is in $\mathcal{O}(m)$. However, it was left open whether a similar bound holds for RMM or not. We show that the answer is negative by providing an explicit graph construction and coloring for which the stabilization time of RMM is exponential, in $n$. Furthermore, we investigate the stabilization time when the underlying graph is a cycle $C_n$. We prove the upper bound of $\lceil n/2 \rceil-1$ for MM and $\mathcal{O}(n^2)$ for RMM. For the former we exploit some combinatorial arguments and for the latter we analyze the ``convergence'' time of a corresponding Markov chain. We show that both of these bounds are tight.

A trivial bound on the periodicity of MM is $2^n$. However, Goles and Olivos~\cite{GOLES1980187} proved that its periodicity is one or two, i.e., the process always reaches a fixed coloring or switches between two colorings. While a similar behavior was observed for RMM on some special classes of graphs, cf.~\cite{abdullah2015global}, we prove that this does not apply to the general case. More precisely, we give graph structures and initial colorings for which the periodicity of RMM is exponential.

We also initiate the study of the number of stable colorings. We prove that the number of stable colorings of a cycle $C_n$ is in $\Theta(1)$ for RMM and in $\Theta(\Phi^n)$ for MM, where $\Phi=(1+\sqrt{5})/2$ is the golden ratio. This is another indication how small alterations in the local behavior of a process, such as the tie-breaking rule, can have a substantial impact on the global behavior of the process.

\textbf{Contribution 2: Minimum Size of a Winning Set.} 
We provide some bounds on minimum-size winning sets. In particular, in RMM on a cycle $C_n$, the only winning set is the set of all nodes. In MM on $C_n$, the minimum size of a winning set is equal to $\lfloor n/2\rfloor+1$.


\textbf{Contribution 3: Random Initial Coloring.} The problem of finding the expected ``final'' number of blue nodes starting from a $p$-random coloring has been attacked by previous work (see Section~\ref{prior work}). However, only some loose bounds for special classes of graphs have been provided, which seems to be due the inherent difficulty of the problem. We make some advancements on this front, by answering the question for cycle graphs. (As we explain later, we believe that our techniques can be used to prove similar results for a larger class of graphs, namely the $d$-dimensional torus or more broadly vertex-transitive graphs.) We show that in RMM on $C_n$, the expected final number of blue nodes is equal to $pn$. On the other hand, this is equal to $(2p^2-p^3)n/(1-p+p^2)$ for MM (it was brought to our attention that a similar result was proven in~\cite{mossel2014majority}. However, we believe our proof is more intuitive and more importantly we prove a w.h.p. statement).

\textbf{Contribution 4: Proof Techniques.} One of the main contributions of the present paper is introducing several proof techniques built on Markov chain analysis, counting arguments, potential functions, greedy approaches, martingale processes, and recursive functions, which we believe can be very beneficial for the future work to make advancements on majority based (more generally, threshold based) opinion diffusion models. A fair amount of effort has been put into ensuring that the proofs are accessible by avoiding unnecessary complexities imposed by adding less essential components to the model or the underlying graph structure. This has been our main motive for focusing on two of the most basic models and presenting a big fraction of our results on cycle graphs. We explain how some of our techniques can potentially be utilized to prove similar results in a more general framework.

\textbf{Contribution 5: Experimental Results.} We present the outcomes of several experiments that we have conducted. A subset of these experiments has been designed to merely support and complement our theoretical findings. However, some of the executed experiments let us uncover other interesting characteristics of our models. In particular, we investigate the effect of adding some random edges to the underlying graph structure. This leads to some open problems and conjectures about the connection between graph parameters such as conductance and vertex-transitivity and the process properties such as the stabilization time, which could serve as potential future research directions.

\subsection{Related Work}
\label{prior work}
Numerous opinion diffusion models have been introduced to study how the members of a community form their opinions through social interactions, cf.~\cite{imber2021probabilistic, bara2021predicting}. Among all these models, a considerable amount of attention has been devoted to the study of the majority based models, cf.~\cite{DevilDetails, Auletta_Fanelli_Ferraioli_2019, brill2016pairwise, zehmakan2021majority,amir2023majority}.

\textbf{Stabilization Time and Periodicity.} It was proven~\cite{GOLES1980187} that the periodicity of MM is always one or two. Chistikov et al.~\cite{chistikov2020convergence} showed that it is PSPACE-complete to decide whether the periodicity is one or not for a given coloring of a \emph{directed} graph. Furthermore, it was proven~\cite{poljak1986pre} that the stabilization time of MM is bounded by $\mathcal{O}(m)$. Stronger bounds are known for special classes of graphs. For instance, for a $d$-regular graph with strong conductance the stabilization time is in $\mathcal{O}(\log_d n)$, cf.~\cite{zehmakan2020opinion}. The stabilization properties have also been studied for other majority based models, cf.~\cite{berenbrink2022asynchronous, abdullah2015global,n2020rumor,gartner2020threshold,zehmakan2019spread}. 

\textbf{Minimum Size of a Winning Set.} Motivated from viral marketing where a company aims to trigger a large cascade of further adoptions of its product by convincing a subset of individuals to adopt a positive opinion about its product (e.g., by giving them free samples), the problem of finding the minimum size of a winning set has been studied extensively, cf.~\cite{jeger2019dynamic, auletta2020effectiveness, karia2022hard}. G\"artner and Zehmakan~\cite{gartner2018majority} proved that the minimum size of a winning set in MM on a random $d$-regular graph is almost as large as $n/2$ w.h.p. if $d$ is sufficiently large. Using the expander mixing lemma, it was proven~\cite{zehmakan2020opinion} that this is actually true for all graphs with a certain level of conductance, including random regular graphs and Erd\H{o}s-R\'{e}nyi random graph. For general graphs, it was proven in~\cite{auletta2018reasoning} that every graph has a winning set of size at most $n/2$ under the asynchronous variant of MM. In~\cite{MAIN_PELEG, out2021majority}, the minimum size of a winning set on graph data from real-world social networks was investigated for a variant of MM where the nodes with the highest degrees (called the elites) have a larger ``influence factor'' than others.

Furthermore, the problem of finding the minimum size of a winning set for a given graph $G$ is known to be NP-hard for different majority based models, cf.~\cite{schoenebeck2020limitations, karia2022hard}, and approximation algorithms based on various techniques, such as integer programming formulations~\cite{wilder2017controlling, tao2022hard} and reinforcement learning~\cite{kamarthi2019influence}, have been proposed. For MM and RMM, it was proven~\cite{mishra2002hardness} that this problem cannot be approximated within a factor of $(\log \Delta \log\log \Delta)$, unless P=NP, but there is a polynomial-time $(\log \Delta)$-approximation algorithm, where $\Delta$ is the maximum degree. Chen~\cite{chen2009approximability} proved that the problem is traceable for special classes of graphs such as trees.

\textbf{Random Initial Coloring.} The problem of finding the expected final number of blue nodes in MM and RMM with a $p$-random initial coloring has been studied for different graphs, e.g., random regular graphs~\cite{gartner2018majority}, hypercubes~\cite{balogh2006bootstrap} and preferential attachment graphs~\cite{amin2018phase}. Motivated from applications in certain interacting particle systems such as fluid flow in rocks and dynamics of glasses, this also has been studied extensively when the underlying graph is a $d$-dimensional torus, cf.\cite{balister2010random,gartner2017color,zehmakan2019two}. Gray~\cite{gray1987behavior} studied the problem for cycle graphs where some noise is added to the process. Roughly speaking, the main finding of the aforementioned work is that there are thresholds $p_1$ and $p_2$ so that if $p$ is sufficiently smaller than $p_1$ (similarly larger than $p_2$) then the process reaches the white (resp. blue) coloring and a non-monochromatic configuration if $p$ is in between w.h.p. The main difficulty in this set-up is to determine the values of $p_1$ and $p_2$.

In the last few years, a lot of attention has been given to the study of MM on  Erd\H{o}s-R\'{e}nyi random graph starting from a $p$-random initial coloring. In~\cite{zehmakan2020opinion}, it was proven that when $p$ is ``slightly'' larger than $1/2$, then the process reaches the blue coloring w.h.p. Following up on a conjecture from~\cite{CONJECTURE_convergence_benjamini}, the case of $p=1/2$ also has been studied extensively, cf.~\cite{sah2021majority,chakraborti2021majority,tran2020reaching}.

\section{Stabilization Time and Periodicity}
\label{stabilizarion}
\subsection{Stabilization Time in General Graphs}
As mentioned, it was proven~\cite{poljak1986pre} that the stabilization time of MM is in $\mathcal{O}(m)$. It is easy to argue this bound holds even when the nodes are updated asynchronously or when we have a biased tie-breaking rule (i.e., always blue is chosen in case of a tie). However, it was left open whether a similar bound can be proven for random tie-breaking. We settle this, in Theorem~\ref{exp-stable-rmm}, by providing an explicit graph construction and coloring for which RMM needs exponentially many rounds to stabilize in expectation. (Our proof actually works for any random tie-breaking rule, where a node chooses blue (white) independently w.p. $0<q<1$ (resp. $1-q$) in case of a tie.)

\begin{theorem}
\label{exp-stable-rmm}
There is a graph $G=(V,E)$ and a coloring $\mathcal{C}_0$ for which the stabilization time of RMM is exponential in $n$.
\end{theorem}
\begin{proof}
To provide the construction of graph $G$, we first define three smaller graphs and then explain how to connect these graphs to create $G$. We define $\kappa := \lfloor n/3 \rfloor-1$. Let $S_b$ be a star graph with an internal node $v_b$ and $\kappa-1$ leaves and $S_w$ be a star graph with an internal node $v_w$ and $n-2\kappa-1$ leaves. Furthermore, let $I$ be the graph built of $\kappa$ isolated nodes. Now to build graph $G$, for each node in $I$ we add an edge to $v_b$ and an edge to $v_w$. (Note that the total number of nodes is equal to $|V_{S_b}|+|V_{S_w}|+|V_I|=\kappa+(n-2\kappa)+\kappa=n$.) Please see Figure~\ref{figure} (left) for an example.

\begin{figure}[t]
  \centering
  \includegraphics[width=1\linewidth]{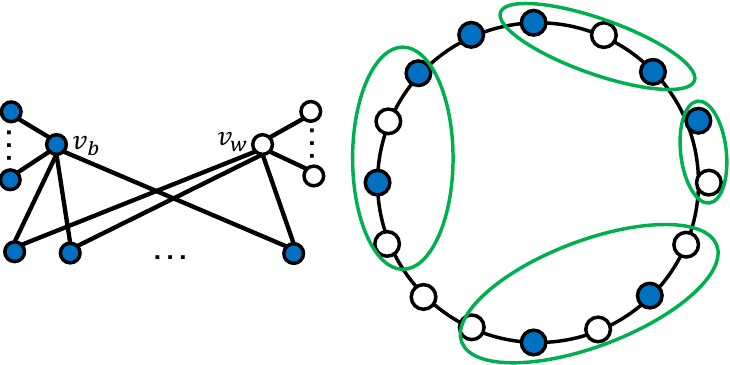}
  \caption{(left) The construction given in Theorem~\ref{exp-stable-rmm} for exponential stabilization time in RMM. (right) The set of ``extended'' maximal alternating paths $\mathcal{A}^{+}$ are enclosed with green curves, see proof of Theorem~\ref{rmm-stabilization-thm}.}
  \label{figure}
\end{figure}

\textbf{Claim 1.} \textit{The nodes in $S_w$ form a resilient set.} Each node in $S_w$ has more than half of its neighbors in $S_w$. This is trivial for all the leaf nodes. The internal node $v_w$ is adjacent to $n-2\kappa-1$ leaves in $S_w$ and $\kappa$ nodes in $I$ and we have $n-2\kappa-1 > \kappa$.

\textbf{Claim 2.} \textit{Let $\mathcal{U}$ be the set of colorings where $S_w$ is white, $S_b$ is blue, and at least one node in $I$ is blue. For a coloring $\mathcal{C}\in \mathcal{U}$, in the next round, all nodes in $S_b$ and $S_w$ keep their color and each node in $I$ chooses a color uniformly at random.} All nodes in $S_w$ remain white according to Claim 1. All leaves in $S_b$ have exactly one neighbor which is blue; thus, they remain blue. Node $v_b$ is of degree $2\kappa-1$ and has at least $\kappa$ blue neighbors, thus it remains blue too. Each node in $I$ has exactly one blue neighbor ($v_b$) and one white neighbor ($v_w$), thus it chooses among blue and white uniformly at random.

Assume that in $\mathcal{C}_0$, all nodes in $S_w$ are white and the rest of the nodes are blue. $\mathcal{C}_0$ is clearly in $\mathcal{U}$. We show that the process eventually reaches the white coloring. Hence, the stabilization time is upper-bounded by the expected number of rounds we need to reach a coloring not in $\mathcal{U}$ (because the white coloring obviously is not in $\mathcal{U}$). Note that from a coloring in $\mathcal{U}$, if at least one node in $I$ selects blue, we are still in $\mathcal{U}$ in the next round, according to Claim 2. The only way to leave $\mathcal{U}$ is that all nodes in $I$ select white. Since this happens only w.p. $1/2^{\kappa}$, it takes $2^{\kappa}=2^{\lfloor n/3\rfloor-1}$ rounds in expectation for it to happen.

It remains to prove that the process eventually reaches the white coloring. Note that according to Claim 1, $S_w$ remains white forever. Thus, it suffices to prove that from any coloring where $S_w$ is fully white, there is a non-zero probability to reach the white coloring. Let $\mathcal{C}$ be such a coloring. There is a non-zero probability that all nodes in $I$
 become white in the next round (since they all have at least one white neighbor, namely $v_w$). It is possible that in the round after all nodes in $I$ remain white and $v_b$ becomes white (recall $d(v_b)=2\kappa-1$). One round after that, all nodes will be white.

\end{proof}

\subsection{Stabilization Time in Cycles}

We prove that on a cycle $C_n$ the stabilization time is at most $\lceil n/2\rceil-1$ for MM (Theorems~\ref{mm-stabilization-thm}) and in $\mathcal{O}(n^2)$ for RMM (Theorem~\ref{rmm-stabilization-thm}). It is straightforward to infer Theorem~\ref{mm-stabilization-thm} from Lemma~\ref{alternating-lemma}, given below. However, for the sake of completeness, we provide a proof for Theorem~\ref{mm-stabilization-thm} in the appendix, Section~\ref{mm-stabilization-thm-appendix}. Furthermore, to prove Theorem~\ref{rmm-stabilization-thm}, we rely on the Markov chain analysis given in Lemma~\ref{markov-chain-lemma} whose full proof is presented in the appendix, Section~\ref{markov-chain-lemma-appendix}.

\begin{lemma}
\label{alternating-lemma}
In MM on a cycle $C_n$ with a coloring $\mathcal{C}$, if there exist two adjacent nodes with the same color, the process reaches a stable coloring after exactly $\lceil l/2 \rceil$ rounds, where $l$ is the length of the longest alternating path in the path partition of $\mathcal{C}$.
\end{lemma}
\begin{proof}
Let $B$ (resp. $W$) be the set of nodes on the (maximal) blue (resp. white) paths in the path partition in $\mathcal{C}$. All nodes in $B$ and $W$ keep their color forever. Furthermore, all alternating paths in the path partition keep shrinking until they disappear. Consider an alternating path $v_1,\cdots, v_k$. After one round, $v_1$ and $v_k$ ``join'' the adjacent monochromatic paths and thus it shrinks to the alternating path $v_2,\cdots, v_{k-1}$, which is of length $k-2$. If $k$ is even, the path disappears after $k/2=\lceil k/2 \rceil$ rounds. If $k$ is odd, its length decreases by two in each round until it is of length $1$. Then, it needs one more round to disappear. This is equal to $\lceil k/2\rceil$ rounds overall. Therefore, after $\lceil l/2 \rceil$ rounds all nodes are on a monochromatic path of length at least two and will never change their color.
\end{proof}

\begin{theorem}
\label{mm-stabilization-thm}
The stabilization time of MM on a cycle $C_n$ is at most $\lceil n/2\rceil -1$ and this bound is tight.
\end{theorem}

\begin{lemma}
\label{markov-chain-lemma}
Consider the time-homogenous Markov chain which is defined over the state set $S:=\{s_0,\cdots,s_k\}$ with the transition matrix $P:=(p_{s_i,s_j})_{s_i,s_j\in S}$, where for $1 \le i \le k-1$ we have $p_{s_i,s_i}=\frac{1}{2}$ and $p_{s_i,s_{i+1}}=p_{s_i,s_{i-1}}=\frac{1}{4}$ and for $i=0,k$ we have $p_{s_i,s_i}=1$. The expected number of rounds it needs to reach from a state $s_i$ to $s_0$ or $s_k$ is equal to $2i(k-i)$.
\end{lemma}
\textsc{Proof Sketch.}
Let $T_i$ be the expected number of rounds the Markov chain needs to reach from state $s_i$ to state $s_0$ or $s_k$. Obviously, we have $T_0=T_k=0$. Furthermore, from state $s_i$, for $1\le i\le k-1$, if we move to state $s_{i+1}$ w.p. $1/4$, then in addition to this step we need in expectation $T_{i+1}$ steps to finish. A similar argument applies to the transition to $s_{i-1}$ and remaining in state $s_i$, which happen w.p. $1/4$ and $1/2$ respectively. Thus, conditioning on these three possibilities we conclude that $T_i=\frac{1}{4}T_{i-1}+\frac{1}{4}T_{i+1}+\frac{1}{2}T_{i}+1$ for $1\leq i \leq k-1$.
Solving this linear recursion gives us $T_i=2i(k-i)$. (Please see Section~\ref{markov-chain-lemma-appendix} for a full proof.) \qed

\begin{theorem}
\label{rmm-stabilization-thm}
The stabilization time of RMM on $C_n$ is in $\mathcal{O}(n^2)$.
\end{theorem}
\begin{proof}
Let us first introduce \textit{lazy} RMM on $C_n$ which is basically a slower version of RMM. For a coloring $\mathcal{C}$, consider all the maximal monochromatic paths of length at least 2 on $C_n$, and let $\mathcal{A}$ denote the set of maximal alternating paths which sit between two such monochromatic paths. This includes alternating paths of length 0, when two monochromatic paths with opposite colors are adjacent. (This is essentially the set of alternating paths in the path partition in $\mathcal{C}$ plus the mentioned path of length 0.) Define $\mathcal{A}^{+}$ to be the set of paths obtained by taking each path from $\mathcal{A}$ and attaching its two adjacent nodes to it. (See Figure~\ref{figure} (right) for an example.) In the lazy RMM instead of updating all nodes at once, we pick up the paths in $\mathcal{A}^{+}$ one by one (in an arbitrary order) and then update the color of all nodes on the picked path at once following the RMM rule. Once we have exhausted $\mathcal{A}^{+}$, we regenerate $\mathcal{A}^{+}$ for the new coloring and continue. However, note that we do not actually bring the updated colors to effect until we have gone through all paths in $\mathcal{A}^+$. You can imagine that we keep the updated color for each node in a buffer and then it comes to effect once $\mathcal{A}^+$ is empty.

Note that every two paths in $\mathcal{A}^+$ are disjoint (because we considered the monochromatic paths of length at least two). Furthermore, each node not on any path in $\mathcal{A}^+$ will not change its color in RMM since it has the same color as both its neighbors. Thus, the coloring which is generated after processing all elements of $\mathcal{A}^{+}$ is the same as the coloring which would have been outputted had we applied RMM instead (of course, assuming the same source of randomness, i.e., a node makes the same random choice in both processes in case of a tie). Moreover, the lazy RMM stops when the process reaches a coloring where $\mathcal{A}^+$ is empty. This means the process has reached a monochromatic/blinking configuration, which is equivalent to stabilization in RMM, as we prove formally in Theorem~\ref{cycle-periodicity-mm}. In short, the lazy RMM is just a slower version of RMM, where we break a round into smaller sub-rounds. Thus, it suffices to prove our desired upper-bound of $\mathcal{O}(n^2)$ for the lazy RMM.

Let $P := v_1, \cdots, v_k$ be a path in $\mathcal{A}^+$. We claim that after updating the nodes on $P$, the number of blue nodes increases (decreases) by 1 w.p. $1/4$ and remains the same w.p. $1/2$. First consider the case of even $k$. Since the original alternating path $v_2,\cdots, v_{k-1}$ is of even length, the adjacent monochromatic paths containing $v_1$ and $v_k$ must be of opposite colors. Without loss of generality, assume that $v_1$ is blue and $v_k$ is white. Thus for $2\le i\le k-1$, $v_i$ is white for even $i$ and blue for odd $i$. Overall, there are $k/2$ blue nodes before the update. After the update: (i) each node $v_i$, for $2\le i\le k-1$, deterministically switches its color, which gives $(k-2)/2$ blue nodes (ii) $v_1$ and $v_k$ choose a color uniformly and independently at random. They both choose blue (white) w.p. $1/4$, which gives $(k-2)/2+2=k/2+1$ (resp. $(k-2)/2=k/2-1$) blue nodes, i.e., an increase (resp. decrease) by one in the number of blue nodes. Furthermore, one of them chooses blue and the other one chooses white w.p. $1/2$ which gives $(k-2)/2+1=k/2$ blue nodes, i.e., no change. We can prove the same statement for the case of odd $k$ by applying a very similar argument.


Consider the Markov chain described in Lemma~\ref{markov-chain-lemma} for $k=n$, where state $s_i$ represents having $i$ blue nodes. We claim that the maximum number of rounds this Markov chain needs to reach $s_0$ or $s_n$, in expectation, is an upper bound on the stabilization time of the lazy RMM process. As we discussed in each round of the lazy RMM, the number of blue nodes decreases/increases by 1 w.p. $1/4$ and remains the same w.p. $1/2$. For odd $n$, if the process has not reached the white or blue coloring (corresponding to state $s_0$ and $s_n$ in the Markov chain), the set $\mathcal{A}^+$ is non-empty. Thus the Markov chain actually models the lazy RMM precisely. When $n$ is even, it is possible that we reach a coloring where $\mathcal{A}^+$ is empty but we are not in the blue or white coloring (this happens if the process reaches the blinking configuration, where the corresponding Markov chain is in the state $s_{n/2}$). However, as we are looking for an upper bound, this is not an issue. Hence, starting from a coloring with $i$ blue nodes, the stabilization time is bounded by $2i(n-i)$ rounds. Since $2i(n-i)$ is maximized for $i=n/2$, this is at most $n^2/2=\mathcal{O}(n^2)$.

\end{proof}

The quadratic bound given in Theorem~\ref{rmm-stabilization-thm} is tight. In the appendix, Section~\ref{rmm-stabilization-thm-appendix}, we prove that if we start from a coloring which partitions the node set into a blue path and an alternating path (both of size almost $n/2$) then the process needs $\Omega(n^2)$ rounds in expectation to stabilize.

\subsection{Periodicity in General Graphs and Cycles}
A trivial upper bound on the periodicity of MM and RMM is $2^n$. It was proven~\cite{GOLES1980187} that the periodicity of MM is always 1 or 2. Theorem~\ref{thm-exp-period} states that for RMM the trivial bound of $2^n$ is actually the best possible, up to some constant factor. On the other hand, if we limit ourselves to the cycle graphs, then the periodicity for both RMM and MM is always one or two, see Theorems~\ref{cycle-periodicity-mm}.

\begin{theorem}
\label{thm-exp-period}
For any integer $n$, there is an $n$-node graph $G$ for which the periodicity of RMM is in $\Omega(2^n)$.
\end{theorem}
\textsc{Proof Sketch.}
Define $\kappa$ to be the largest integer smaller than $n-6$ which is divisible by 4. Consider a path $P:= v_0,\cdots, v_{\kappa-1}$, a clique $C_w$ of size 3, and a clique $C_b$ of size $n-3-\kappa$. To build graph $G$, add an edge between $v_0$ and a node in $C_w$ and an edge between $v_{\kappa-1}$ and a node in $C_b$. Let $\mathcal{U}$ be the set of all colorings where $C_w$ is fully white and $C_b$ is fully blue. Note that $|\mathcal{U}|=2^{\kappa}=\Omega(2^n)$. We can prove that for every two colorings $\mathcal{C},\mathcal{C}'\in \mathcal{U}$, there is a non-zero probability to reach from $\mathcal{C}$ to $\mathcal{C}'$ and there is no transition possible from a coloring in $\mathcal{U}$ to a coloring outside $\mathcal{U}$. Thus, the colorings in $\mathcal{U}$ form an absorbing strongly connected component, which yields the bound of $\Omega(2^n)$ on the periodicity. Please refer to the appendix, Section~\ref{thm-exp-period-appendix}, for a full proof.
\qed

\begin{theorem}
\label{cycle-periodicity-mm}
In MM on a cycle $C_n$:
\begin{itemize}
  \item If $n$ is odd, the process always reaches a stable coloring.
  \item If $n$ is even, the process reaches a stable coloring or the blinking configuration.
\end{itemize}
In RMM on a cycle $C_n$:
\begin{itemize}
  \item If $n$ is odd, the process always reaches the white (blue) coloring.
  \item If $n$ is even, the process reaches the white (blue) coloring or the blinking configuration.
\end{itemize}
\end{theorem}

\textsc{Proof Sketch.}
For MM, if there are two adjacent nodes with the same color, then according to Lemma~\ref{alternating-lemma}, the process reaches a stable coloring. If not (which is only possible for even $n$), then the process is in the blinking configuration. For RMM, we need a similar, but probabilistic, argument. The full proof is given in the appendix, Section~\ref{cycle-periodicity-mm-appendix}. \qed

\textbf{Number of Stable Colorings.} According to Theorem~\ref{cycle-periodicity-mm}, there are two stable colorings, namely the white and blue coloring, in RMM on cycle $C_n$. What about the number of stable colorings in MM? We answer this question in Theorem~\ref{stable-colorings-mm}, whose proof is given in the appendix, Section~\ref{appendix-stable-colorings-mm}.

\begin{theorem}
\label{stable-colorings-mm}
In MM on a cycle $C_n = (v_0, \cdots, v_{n-1})$, there are $\Theta(\Phi^n)$ stable colorings, where $\Phi=\frac{1+\sqrt{5}}{2}$ is the golden ratio.
\end{theorem}
\section{Winning Sets}
\label{winning set}
How small could a winning set be? Berger~\cite{berger2001dynamic}, surprisingly, proved that there exist arbitrarily large graphs which have winning sets of constant-size in MM. Actually, a proof was sketched that this statement holds regardless of the tie-breaking rule. This is stated more formally in Theorem~\ref{min-winning-rmm} and for the sake of completeness a full proof is given in  the appendix, Section~\ref{min-winning-rmm-appendix}.

We say a model follows the majority rule if in each round, every node updates its color to the most frequent color in its neighborhood, and a tie is broken in any \textit{arbitrary} manner. This in particular includes MM and RMM.

\begin{theorem}
\label{min-winning-rmm}
For every positive integer $k$ and a model which follows the majority rule, there is an $n$-node graph with $n\ge k$, which has a winning set of size $36$.
\end{theorem}

\begin{theorem}
\label{cycle-winning-set}
In RMM on a cycle $C_n = (v_0, \cdots, v_{n-1})$, the only winning set is the set of all nodes. In MM on $C_n$, the minimum size of a winning set is equal to $\lfloor n/2\rfloor+1$.
\end{theorem}
\textsc{Proof Sketch.}
For RMM, we can prove that if there is a white node in the initial coloring, it is possible that the process does not reach the blue coloring. This implies that the only winning set is the set of all nodes.

Let $\mathcal{B}$ be a winning set in MM. For every two adjacent nodes, at least one must be in $\mathcal{B}$. By a case distinction between odd and even $n$, we can conclude that $|\mathcal{B}|\ge \lfloor n/2\rfloor+1$. Furthermore, this bound is tight since the set $\{v_i: (i \mod 2) =1\}\cup \{v_0\}$ is a winning set of size $\lfloor n/2\rfloor+1$. A full proof is given in the appendix, Section~\ref{cycle-winning-set-appendix}.
\qed

\section{Random Initial Coloring}
\label{random initial}

We determine the expected final number of blue nodes starting with a random coloring on a cycle graph for MM and RMM respectively in Theorems~\ref{random-mm} and \ref{random-rmm}. 

\begin{theorem}
\label{random-mm}
In MM on a cycle $C_n$ with a $p$-random initial coloring for some $p\ge 1/2$, the process reaches a stable coloring with $(1\pm \epsilon)\frac{2p^2-p^3}{1-p+p^2}n$ blue nodes, for an arbitrarily small constant $\epsilon>0$, in $\mathcal{O}\left(\log n\right)$ rounds w.h.p.
\end{theorem}
\begin{proof}
Let $\mathcal{E}$ be the event that there is no alternating path of size larger than $n-4$ in the initial coloring. The probability of $\mathcal{E}$ not happening can be upper-bounded by $2n(p(1-p))^{\lfloor(n-4)/2\rfloor}$ which is exponentially small in $n$. Since our statement needs to hold w.h.p. (i.e., w.p. $1-o(1)$), in the rest of the proof, we assume that $\mathcal{E}$ happens. (To be fully accurate, we need to condition on $\mathcal{E}$ happening in our calculations, but we skip that for the sake of simplicity.) Thus, in the initial coloring the nodes can be partitioned into maximal blue and white paths of length at least two and maximal alternating paths of size at most $n-4$. From such an initial coloring, the monochromatic paths keep growing and the alternating paths shrink until the process reaches a stable coloring with only monochromatic paths. (See proof of Lemma~\ref{alternating-lemma} for more details.)

Let $p_f$ be the probability that an arbitrary node $v$ is blue at the end. To compute $p_f$, we consider the three cases of $v$ being on a monochromatic path, on an odd alternating path, or an even alternating path in the path partition of the initial coloring, which results in Equation~(\ref{eq-pf}). (I) If $v$ is on a white path, it never becomes blue. If it is on a blue path, it remains blue forever. The probability of $v$ being on a blue path is equal to $p(p^2+2p(1-p))$ since $v$ and at least one of its neighbors must be blue. (See the first term in Equation~(\ref{eq-pf}).) (II) An odd alternating path is adjacent to two monochromatic paths of the same color (they potentially could be the same path) and all nodes on the alternating path eventually choose the color of the monochromatic path(s). The probability that $v$ is on an odd alternating path of length $k$ which is adjacent to blue path(s) is equal to $p^4kp^{\lfloor k/2\rfloor}(1-p)^{\lceil k/2\rceil}$. (The term $p^4$ is for two adjacent nodes at each side of the path to be blue. Note that since we assume that there is no alternating path of size larger than $n-4$, these four nodes are distinct.) Summing over all choices of $k$, we get the second term in Equation~(\ref{eq-pf}). (III) An even alternating path $P$ is adjacent to a blue path and a white path. The nodes on $P$ which are closer to the blue (white) path become blue (resp. white) after at most $|P|/2$ rounds. The probability that $v$ is on an alternating even path of length $k$ and is closer to the blue path is equal to $2p^2(1-p)^2\frac{k}{2}p^{k/2}(1-p)^{k/2}$. Summing over all choices of $k$, we get the third term in Equation~(\ref{eq-pf}).

\begin{equation}
\begin{split}
    \label{eq-pf}
    p_f =& \left(2p^2-p^3\right)+ p^4\sum_{\textrm{odd } 1\le k\le n-4}k p^{\bigl\lfloor \frac{k}{2}\bigr\rfloor}(1-p)^{\bigl\lceil \frac{k}{2}\bigr\rceil}+\\ &p^2(1-p)^2\sum_{\textrm{even } 1\le k\le n-4}k p^{\frac{k}{2}}(1-p)^{\frac{k}{2}}
    \end{split}
\end{equation}

Let us define $q:= p(1-p)$. Then we can write the last sum as $2q^2\sum_{i=1}^{\lfloor \frac{n-4}{2}\rfloor}iq^i$. This is equal to $2q^2\frac{q}{(1-q)^2}+\mathcal{O}(nq^{\frac{n}{2}})$, where we used the fact that this is the derivative of a geometric series. Similarly, we can show that the first sum in Equation~(\ref{eq-pf}) is equal to $p^3(\frac{2q}{(1-q)^2}-\frac{q}{1-q})+\mathcal{O}(nq^{\frac{n}{2}})$. By plugging these into Equation (\ref{eq-pf}), doing some basic calculations, and using the fact that $n$ tends to infinity, we get $p_f=\frac{2p^2-p^3}{1-p+p^2}$. (We are ignoring the additive term $\mathcal{O}(nq^{\frac{n}{2}})$ because it is converging to 0 and  can be hidden behind the estimate $(1\pm \epsilon)$ that we add later.) This implies that $\mathbb{E}[b_f]=\frac{2p^2-p^3}{1-p+p^2}n$ where $b_f$ is the final number of blue nodes.

Let $l_p$ denote the length of the longest alternating path in a $p$-random coloring on $C_n$. Then, for $l^* :=8\log_2 n$ we have
\begin{equation}
\label{eq-l*}
    \mathbb{P}[l_p\ge l^*]\le 2n (p(1-p))^{l^*/2}\le 2n \left(\frac{1}{2}\right)^{l^*/2}=\frac{2}{n^3}.
\end{equation}
Therefore, w.p. at least $1-2/n^3$, the process ends before $l^*$ rounds.

We claim that the random variable $b_f$ (defined over $\Omega =\{w,b\}^n$) is difference-bounded by $(\beta=n,c=4l^*+7,\delta=2/n^3)$. (I) Let $B$ be the set of colorings where there is an alternating path of length at least $l^*$. If we set $p=1/2$, then we pick a coloring uniformly at random among the $2^n$ colorings. According to Equation~(\ref{eq-l*}), the probability that such a randomly chosen coloring has an alternating path of size at least $l^*$ is at most $2/n^3$, i.e., $|B|/|\Omega|\le 2/n^3$. (II) Consider a coloring $\mathcal{C}\notin B$. Since the length of the longest alternating path is less than $l^*$, the process ends before $l^*$ rounds. Now, assume we flip the color of a node $v$ to obtain the coloring $\mathcal{C}'$. The longest alternating path in $\mathcal{C}'$ cannot be longer than $2l^*+3$. Thus, the process starting from $\mathcal{C}'$ ends in at most $t\le 2l^*+3$ rounds. Furthermore, the color of node $v$ influences the final color of at most $2t+1$ nodes, namely the nodes whose distance from $v$ is at least $t$. Therefore, the difference between the final number of blue nodes when starting from $\mathcal{C}$ and $\mathcal{C}'$ is at most $2(2l^*+3)+1=4l^*+7$, i.e., $|b_f(\mathcal{C})-b_f(\mathcal{C}')|\le 4l^*+7$. (We are actually quite generous with our calculations here.) (III) For two arbitrary colorings $\mathcal{C}$ and $\mathcal{C}'$, we trivially have $|b_f(\mathcal{C})-b_f(\mathcal{C}')|\le n$. Now, applying Theorem~\ref{McDiarmid-thm} implies that $\mathbb{P}[(1-\epsilon)\mathbb{E}[b_f]\le b_f\le (1+\epsilon)\mathbb{E}[b_f]]$, for some $\epsilon>0$, is at least $1-2\exp(-(\epsilon^2\mathbb{E}[b_f]^2)/(8n(4l^*+7)^2))-4/(n(4l^*+7))$ where we used $\beta=n$, $c=4l^*+7$, $\delta=2/n^3$. Using $\mathbb{E}[b_f]^2=(2p^2-p^3)^2n^2/(1-p+p^2)^2=\Theta(n^2)$ for $p\ge 1/2$ and $(4l^*+7)^2=\Theta(\log^2 n)$, the above probability is at least $1-\exp(-\Theta(n/\log^2 n))-1/\Theta((n\log n))=1-o(1)$. Furthermore, we already proved that the process ends w.h.p. before $8\log_2 n$ rounds. Therefore, the process reaches a stable coloring with $(1\pm \epsilon)\frac{(2p^2-p^3)}{1-p+p^2}n$ blue nodes in $\mathcal{O}(\log n)$ rounds w.h.p.
\end{proof}

\begin{theorem}
\label{random-rmm}
Consider RMM on $C_n$ and assume that $b_0=pn$ for some $0\le p\le 1$. Then, we have $\mathbb{E}[b_t]=pn$ for any $t\in \mathbb{N}$.
\end{theorem}
\begin{proof}
It suffices to prove that the sequence $b_0, b_1, b_2,\cdots$ is a discrete-time martingale, i.e., $\mathbb{E}[b_t|b_0, b_1, \cdots, b_{t-1}]=b_{t-1}$. Let us formulate RMM in a slightly different way. Assume that in each round, a white (blue) node sends a white (blue) pebble to each of its two neighbors. Then, each node uniformly and independently at random chooses one of the two pebbles it has received and picks its color. This is the same as the RMM rule because if the neighbors of a node agree on a color, it picks that color w.p. 1, and otherwise it picks a color independently and uniformly at random. Now, assume that there are $b$ blue nodes in the round $t-1$. Then, each of the $b$ blue nodes sends out two blue pebbles and each blue pebble is selected and results in a blue node w.p. $1/2$. Thus, by the linearity of expectation, the expected number of blue nodes in round $t$ is equal to $2b*(1/2)=b$. This concludes the proof that the sequence is a martingale. Therefore, we have $\mathbb{E}[b_t]=pn$ for any $t\in \mathbb{N}$.
\end{proof}

Theorem~\ref{random-rmm} holds for any initial coloring with $pn$ blue nodes, regardless of their position. We can apply this to the case of a $p$-random initial coloring because a simple application of the Chernoff bound~\cite{dubhashi2009concentration} implies that there are $pn$ blue nodes initially w.h.p. up to some ``small'' error factor.

\begin{figure*}[t]
\includegraphics[scale=0.47]{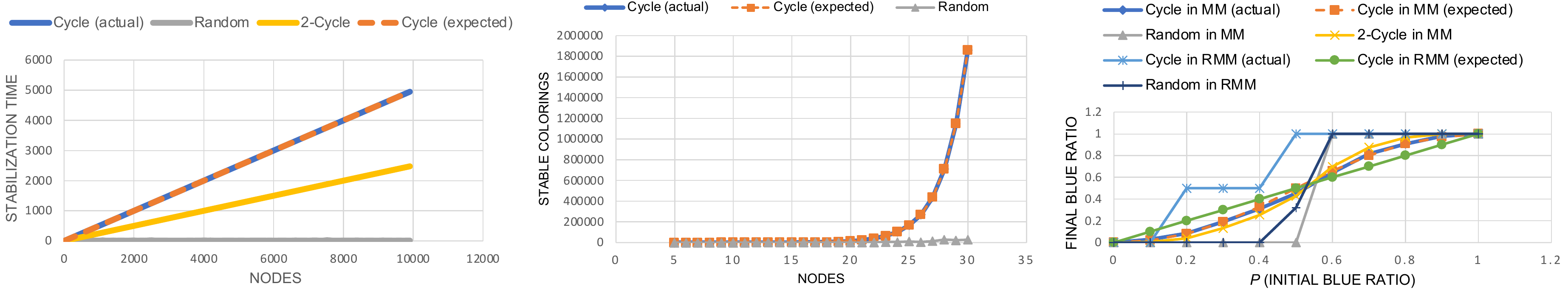}
  \caption{(left) The stabilization time in MM as a function of $n$, with a white path of length 2 (or 3) and an alternating path of length $n-2$ (or $n-3$) (middle) The number of stable colorings as a function of $n$ in MM (right) The final ratio of blue nodes for different values of $p$, starting from a $p$-random coloring. We use ``expected'' for what we expect according to our theoretical findings and ``actual'' is the output of the experiments.}
  \label{fig-exp}
\end{figure*}

\begin{corollary}
\label{corollary}
For RMM on $C_n$ with $b_0=pn$:
\begin{itemize}
    \item If $n$ is odd, the process reaches the blue coloring w.p. $p$ and the white coloring w.p. $1-p$.
    \item if $n$ is even, the process reaches the blue coloring w.p. $p^2$, the white coloring w.p. $(1-p)^2$ and the blinking configuration w.p. $2p(1-p)$.
\end{itemize}
\end{corollary}


\section{Experiments}
\label{experiments}

We also study MM and RMM from an experimental perspective. The conducted experiments not only complement our theoretical findings, but also open doors for future research on the connection between graph characteristics, such as conductance and vertex-transitivity, and the behavior of MM and RMM. Our experiments are executed on cycle, 2-cycle (to build a 2-cycle, take a cycle $C_n$ and add an edge between every two nodes which are in distance 2), and some random graph (which is the graph obtained by adding two randomly selected edges to each node in a cycle $C_n$).

Figure~\ref{fig-exp} (left) depicts the stabilization time of MM on a cycle graph $C_n$ for an extreme coloring, where there is a white path of length $2$ (or $3$) and an alternating path of length $n-2$ (or $n-3$). The stabilization time for the cycle perfectly matches the bound $\lceil n/2 \rceil-1$ proven in Theorem~\ref{mm-stabilization-thm}. Interestingly, once we add two random neighbors for each node on the cycle (to obtain the random graph), then the process ends extremely faster (i.e., in less than 25 rounds even for $n=10,000$). To argue that this is not merely the effect of adding extra edges, but rather how they are added, we ran the process on a 2-cycle graph (which has the same number of edges as the random graph). As you can observe, even though the process speeds up slightly, it is still substantially slower than the random case. Is it true that the stabilization time on random graphs, such as Erd\H{o}s-R\'{e}nyi random graph and random regular graphs, or more generally graphs with strong conductance properties is small, perhaps (sub)-logarithmic in $n$? This is left as an open problem.
(We should mention that our experiments for RMM demonstrated similar behavior change, but they are not included in Figure~\ref{fig-exp}.)

Figure~\ref{fig-exp} (middle) visualizes the number of stable colorings in MM for a cycle $C_n$ obtained from our experiments alongside the expected estimate $\Phi^{n}$ from Theorem~\ref{stable-colorings-mm}. Again, adding random edges results in a considerably different behavior, i.e., the number of stable colorings decreases drastically. Note that a stable coloring corresponds to a partition of the nodes into resilient sets. Thus, if there are many ways to partition a graph's node set into resilient sets, there exist many stable colorings. In graphs with strong conductance properties, such as the above random graph, for the sets which are not too large, the number of edges on the boundary is more than twice the number of edges inside the set. Thus, such sets do not form resilient sets. Another parameter which, we believe, plays a role is vertex-transitivity because it provides a certain level of ``symmetry'' which could result in the formation of resilient sets. Therefore, it would be interesting to characterize the number of stable colorings in terms of different graph parameters, in particular conductance and vertex-transitivity, in the future work.

Figure~\ref{fig-exp} (right) visualizes the final ratio of blue nodes by starting from a $p$-random coloring for different values of $p$ and $n=2000$ in both MM and RMM. For MM on $C_n$, the output of our experiments acceptably matches what one would expect according to our result in Theorem~\ref{random-mm}. For RMM, it, unsurprisingly, does not match the expected final density $p$ (see Theorem~\ref{random-rmm}) because we know that the process always reaches a monochromatic coloring or the blinking configuration (see Corollary~\ref{corollary}). (If we let $n$ be odd, e.g. $n=1999$, then it only can get monochromatic.) Once we switch to our random graph, the process exhibits a behavior called \textit{perfect classification}, i.e., if $p$ is smaller (larger) than $1/2$, then the process reaches the white (resp. blue) coloring. This is aligned with the results from prior work, cf.~\cite{zehmakan2020opinion}, on the relation between conductance and perfect classification. On the other hand, both cycle and $2$-cycle graphs, up to some degree, exhibit a property known as \textit{fair classification}, i.e., the expected final ratio of blue nodes is ``almost'' equal to their initial ratio $p$.

\section{Conclusion}
We studied two very fundamental majority based opinion diffusion processes. Developing several novel proof techniques, we provided tight bounds on the stabilization time, periodicity, minimum size of a winning set, and the expected final density in these processes.

We proved that the stabilization time and periodicity of RMM can be exponential for some graphs. It would be interesting to characterize graphs for which a polynomial upper bound exists.

We initiated the study of the number of stable colorings and provided tight bounds for the cycle graph in both MM and RMM. A potential future research direction is to determine the graph parameters which govern the number of stable colorings. Building on our experimental findings, we nominated conductance and vertex-transitivity as potential candidates.

It is known by prior work, cf.~\cite{zehmakan2020opinion}, that for perfect classification, it suffices that the graph enjoys strong conductance properties. What are the necessary and sufficient conditions for fair classification?

\newpage
\bibliographystyle{ACM-Reference-Format} 
\bibliography{ref}

\newpage
\appendix
\section{Appendix}

\subsection{Proof of Theorem~\ref{mm-stabilization-thm}}
\label{mm-stabilization-thm-appendix}
Let's first consider the case where there are no two monochromatic adjacent nodes in the initial coloring. This is possible only for even $n$. In that case, the process keeps switching between the two alternating colorings, i.e., the process has reached the blinking configuration. In this set-up, the stabilization time is zero by definition.

Now, assume that there are two adjacent monochromatic nodes. Then, the longest alternating path in the path partition is of size at most $n-2$. Thus, according to Lemma~\ref{alternating-lemma}, the process ends after at most $\lceil (n-2)/2\rceil= \lceil n/2 \rceil -1$ rounds.

To prove the tightness, for odd (even) $n$ consider a coloring where two (three) adjacent nodes are white and the remaining nodes form an alternative path of length $n-2$ (resp. $n-3$). According to Lemma~\ref{alternating-lemma}, the process needs $\lceil (n-2)/2\rceil$ rounds, for odd $n$, and $\lceil(n-3)/2\rceil$ rounds, for even $n$, to end. We observe that both these values are equal to $\lceil n/2\rceil -1$. Thus, the bound is tight.

\subsection{Proof of Lemma~\ref{markov-chain-lemma}}
\label{markov-chain-lemma-appendix}
Let $T_i$ be the expected number of rounds the Markov chain needs to reach from state $s_i$ to state $s_0$ or $s_k$. Obviously, we have $T_0=T_k=0$. Furthermore, from state $s_i$, for $1\le i\le k-1$, if we move to state $s_{i+1}$ w.p. $1/4$, then in addition to this step we need in expectation $T_{i+1}$ steps to finish. A similar argument applies to the transition to $s_{i-1}$ and remaining in state $s_i$, which happen w.p. $1/4$ and $1/2$ respectively. Thus, conditioning on these three possibilities we conclude that $T_i=\frac{1}{4}T_{i-1}+\frac{1}{4}T_{i+1}+\frac{1}{2}T_i+1$ for $1\leq i \leq k-1$. By rearranging the terms we get the following non-homogenous linear recursion of order 2:
\begin{align*}
    \frac{1}{2}T_{i+1}-T_i+\frac{1}{2}T_{i-1}=-2, \quad T_0=T_k=0.
\end{align*}
Let us first look at the homogeneous equation $\frac{1}{2}T_{i+1}-T_i+\frac{1}{2}T_{i-1}=0$ whose characteristic equation is equal to $\frac{1}{2}\lambda^2-\lambda+\frac{1}{2}=0$, for some value $\lambda$ to be determined. The characteristic equation has the repeated root $\lambda=1$. Thus, the general solution is of the form $T_i=A+Bi$ for some constants $A$ and $B$.

Now, we need to find a ``particular solution'' to the inhomogeneous equation. If we plug in $Ci^2$ for a constant $C$, we get:
\begin{align*}
    -2 = \frac{1}{2}C(i+1)^2-Ci^2+\frac{1}{2}C(i-1)^2 = C.
\end{align*}
So the general solution to the inhomogeneous equation is equal to $T_i=-2i^2+A+Bi$. Since $T_0=0$ and $T_0=-2*0^2+A+B*0=A$, we have $A=0$. Furthermore, $T_k=0$ and $T_k=-2k^2+A+Bk=-2k^2+Bk$ imply that $B=2k$. Therefore, we can conclude that $T_i=-2i^2+0+2ki=2i(k-i)$.

\subsection{Tightness of Theorem~\ref{rmm-stabilization-thm}}
\label{rmm-stabilization-thm-appendix}

Let $l=n-5$ for even $n$ and $l=n-4$ for odd $n$. (Note that $l$ is odd.)
We define a \textit{$k$-alternating coloring} to be a coloring with a blue path of length $n-k$ plus an alternating path of length $k$ for some odd $k$ between $5$ and $l$. (The alternating path starts and ends with a white node.) Consider a $k$-alternating coloring for $7\le k\le l-2$. Using an argument similar to the one from the proof of Theorem~\ref{rmm-stabilization-thm}, we can observe that from such coloring in the next round, we have a $(k+2)$-alternating coloring (similarly a $(k-2)$-alternating coloring) w.p. $1/4$ and a $k$-alternating coloring w.p. $1/2$.

Let's assume that the process starts from an $l'$-alternating coloring for $l'$ being the closest odd integer to $l/2$. Suppose that we say the process has stabilized if it reaches a $5$-alternating coloring or an $l$-alternating coloring. Note that this is obviously a lower bound on the original stabilization time since for the process to stabilize (i.e., reach a white/blue/blinking configuration, according to Theorem~\ref{cycle-periodicity-mm}) it must first reach one of these two colorings. Therefore, the defined process (running RMM starting from an $l'$-alternating coloring and stopping once reached a $5$-alternating or an $l$-alternating coloring) is equivalent to the Markov chain defined in Lemma~\ref{markov-chain-lemma}, where $s_i$, for $0\le i \le k=(l-5)/2$, corresponds to being in a $(2i+5)$-alternating coloring; in particular, $s_0$ and $s_k$ correspond to being in a $5$-alternating and an $l$-alternating coloring. According to Lemma~\ref{markov-chain-lemma}, the number of rounds to reach a $5$-alternating or an $l$-alternating coloring is $2*\frac{l'-5}{2}\left(\frac{l-5}{2}-\frac{l'-5}{2}\right)$. Using the fact that $l'$ is equal to $l/2\pm 1/2$ and $l\ge n-5$, it is straightforward to show that this is in $\Omega(n^2)$.


\subsection{Proof of Theorem~\ref{thm-exp-period}}
\label{thm-exp-period-appendix}

We define $\kappa$ to be the largest integer smaller than $n-6$ which is divisible by 4. Let us explain how to construct the graph $G$ step by step. Consider a path $P:= v_0,\cdots, v_{\kappa-1}$, a clique $C_w$ of size 3, and a clique $C_b$ of size $n-3-\kappa$. To build the graph $G$, add an edge between $v_0$ and a node in $C_w$, called $u_w$, and an edge between $v_{\kappa-1}$ and a node in $C_b$, called $u_b$. (Note that the output graph has exactly $n$ nodes.) See Figure~\ref{figure-appendix} for an example.

\begin{figure}[h]
  \centering
  \includegraphics[width=0.9\linewidth]{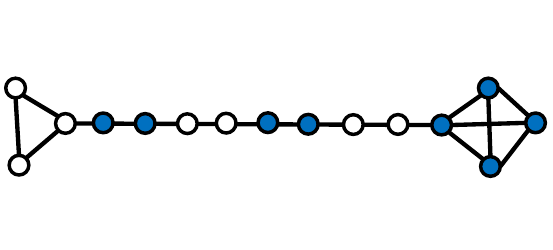}
  \caption{The construction with exponential periodicity in RMM for $n=15$.}
  \label{figure-appendix}
  \Description{}
\end{figure}
Let us observe that $C_w$ (analogously, $C_b$) is a resilient set. Each node $u$ in $C_w$ (analogously, $C_b$) has more than half of its neighbors in $C_w$ (resp. $C_b$). This is true for $u_w$ (resp. $u_b$) since it has 3 neighbors (resp. $n-\kappa-3\ge 3$) neighbors and only one of them is not in $C_w$ (resp. $C_b$). This is trivial for other nodes since they have all their neighbors in $C_w$ (resp. $C_b$).

Let $\mathcal{U}$ be the set of all colorings where $C_w$ is fully white and $C_b$ is fully blue. Note that $|\mathcal{U}|=2^{\kappa}=\Omega(2^n)$. We will prove that for every two colorings $\mathcal{C},\mathcal{C}'\in \mathcal{U}$, there is a non-zero probability to reach from $\mathcal{C}$ to $\mathcal{C}'$, i.e., there is a path from $\mathcal{C}$ to $\mathcal{C}'$ in the underlying directed graph of the corresponding Markov chain. Thus, the colorings in $\mathcal{U}$ form a strongly connected component. Note that there is no edge from a coloring in $\mathcal{U}$ to a coloring outside $\mathcal{U}$ because this requires that a node in $C_w$ or $C_b$ to change its color, which is not possible since they are both resilient sets. Therefore, this is actually an absorbing strongly connected component which implies that the periodicity is in $\Omega(2^n)$.

It remains to prove that for $\mathcal{C},\mathcal{C}'\in \mathcal{U}$, we can reach from $\mathcal{C}$ to $\mathcal{C}'$. Let us define coloring $\mathcal{C}_M\in \mathcal{U}$ where $v_i$ is blue if $(i\mod 4) = 0, 1$ and white otherwise. (See Figure~\ref{figure-appendix} for an example.) In $\mathcal{C}_M$, each node on $P$ has one blue neighbor and one white neighbor and thus chooses its color at random. This implies that we can reach any coloring in $\mathcal{U}$ from $\mathcal{C}_M$. Hence, it suffices to show that there is a path from each coloring $\mathcal{C}\in \mathcal{U}$ to $\mathcal{C}_M$. Firstly, from $\mathcal{C}$ we can reach the coloring where $P$ is fully blue. In the first round, we can color $v_{\kappa-1}$ blue since it has at least one blue neighbor, namely $u_b$. Then, we can color $v_{\kappa-2}$ blue (while $v_{\kappa-1}$ remains blue) since it has at least one blue neighbor, namely $v_{\kappa-1}$, and so on. Thus, after $\kappa$ rounds, $P$ is fully blue. Now, we argue that there is a non-zero probability that in the next four rounds the following updates take place: (i) $v_0$ become white (which is possible since the adjacent node $u_w$ is white) (ii) $v_0$ becomes blue (which is possible since $v_1$ is blue) and $v_1$ becomes white (which is possible since $v_0$ is white) (iii) $v_2$ becomes white and $v_1$ becomes blue (iv) $v_3$ becomes white. (Note that we assume any node which is not mentioned remains unchanged. This is possible since all other nodes have at least one adjacent node of the same color.) After these four rounds, $v_0,v_1$ are blue and $v_2,v_3$ are white, which is identical to their coloring in $\mathcal{C}_M$. Now, we repeat the same process for $v_4, v_5, v_6, v_7$ and so on. After $\kappa/4$ repetitions (i.e., $\kappa$ rounds) we reach $\mathcal{C}_M$. Overall, we can conclude that there is a non-zero probability to reach $\mathcal{C}_M$ from any coloring in $\mathcal{U}$. This finishes the proof.
\subsection{Proof of Theorem~\ref{cycle-periodicity-mm}}
\label{cycle-periodicity-mm-appendix}

First consider MM. If $n$ is odd, then for any coloring there must exist at least two adjacent monochromatic nodes. Thus, according to Lemma~\ref{alternating-lemma}, the process must reach a stable coloring. For even $n$, if there are two adjacent monochromatic nodes, then again we can apply the same argument. If not, then the process is in the blinking configuration.

Now, consider RMM. Let $n$ be odd. It suffices to prove that it is possible (i.e., there is a non-zero probability) to reach from any coloring to the white or blue coloring. Consider an arbitrary coloring $\mathcal{C}$. Since $n$ is odd, there must be two adjacent monochromatic nodes in $\mathcal{C}$. Thus, there exists a monochromatic, say blue, path $P$ of length at least 2. There is a non-zero probability that all nodes on $P$ remain blue in the next round and the node(s) adjacent to $P$ become/stay blue (because all these nodes have at least one blue neighbor). Thus, it is possible that path $P$ keeps growing until it takes over the whole cycle and we reach the blue coloring. For even $n$, if there is at least one monochromatic path of length two or larger, then the above argument applies again. Otherwise, the process is in the blinking configuration.

\subsection{Proof of Theorem~\ref{stable-colorings-mm}}
\label{appendix-stable-colorings-mm}

We say a blue (resp. white) node is \textit{solitary} if both of its neighbors are white (resp. blue). A coloring is stable in MM if and only if it has no solitary node. If there is no solitary node, then each node has a neighbor of the same color and keeps its color, i.e., the coloring is stable. If there is a solitary node in the coloring, it changes its color in the next round, i.e., the coloring is not stable. Thus, we want to determine $|\mathcal{S}_n|$, where $\mathcal{S}_n$ is the set of all colorings on a cycle $C_n$ with no solitary nodes. 

Let $\mathcal{R}_n$ denote the red-green colorings of a cycle $C_n$, where there is an even number of red nodes and there are no two adjacent red nodes. We define a mapping $\mathcal{M}:\mathcal{S}_n\rightarrow \mathcal{R}_n$. $\mathcal{M}$ maps a blue-white coloring $\mathcal{C}\in \mathcal{S}_n$, to a red-green coloring $\mathcal{C}'\in \mathcal{R}_n$ in the following manner: for $0\le i \le n-1$ if $\mathcal{C}(v_i)=\mathcal{C}(v_{i+1})$ ($i+1$ is calculated modular $n$), then $\mathcal{C}'(v_i) = g$ and $\mathcal{C}'(v_i) = r$ otherwise (where $g$ and $r$ stand for green and red). The generated red-green coloring $\mathcal{C}'$ is in $\mathcal{R}_n$ because if there are two adjacent red nodes in $\mathcal{C}'$, then there is a solitary node in $\mathcal{C}$ (but that is not possible since $\mathcal{C}$ is in $\mathcal{S}_n$). Furthermore, since the number of changes from blue to white and white to blue must be even, there is an even number of red nodes. We claim that for each $\mathcal{C}'\in \mathcal{R}_n$, there are exactly two colorings in $\mathcal{S}_n$ which are mapped to $\mathcal{C}'$. Let's try to construct a prospective coloring $\mathcal{C}\in \mathcal{S}_n$ which is mapped to $\mathcal{C}'$. Assume that $\mathcal{C}({v_0})=b$; then $\mathcal{C}(v_i)$, for $1 \le i \le n-1$, is enforced by $\mathcal{C}(v_{i-1})$ and $\mathcal{C}'(v_{i-1})$. For example, if $\mathcal{C}'(v_0)=r$, then $\mathcal{C}(v_1)=w$ (because $v_0$ and $v_1$ must have opposite colors when $\mathcal{C}'(v_0)=r$) and $\mathcal{C}(v_1)=b$ otherwise. Therefore, if we apply the mapping $\mathcal{M}$ on $\mathcal{C}$ we get a coloring which matches $\mathcal{C}'$ on all nodes $v_0, \cdots, v_{n-2}$ by construction. Note that $\mathcal{C}'(v_{n-1})$ must be the same since there are an even number of red nodes. We can construct another coloring which also gets mapped to $\mathcal{C}'$ by starting to color $v_0$ with white. Overall, we argued each coloring in $\mathcal{S}_n$ is mapped to exactly one coloring in $\mathcal{R}_n$ and for each coloring $\mathcal{C}'\in \mathcal{R}_n$, there are exactly two colorings in $\mathcal{S}_n$ which are mapped to $\mathcal{C'}$. This implies that $|\mathcal{S}_n|=2|\mathcal{R}_n|$.

To calculate $r(n):=|\mathcal{R}_n|$, let us first calculate $p(n)$ which is the number of red-green paths of length $n$ with no two adjacent red nodes and an even number of red nodes. It is straightforward to compute the stating values $p(1)$, $p(2)$, $p(3)$, and $p(4)$. Furthermore, we have $p(n)=p(n-1)+p(n-4)+\cdots+p(1)+2$ for $n\ge 5$. This is true because if for an $n$-node path $v_0\cdots, v_{n-1}$ we color $v_0$ with green, then there are $p(n-1)$ ways to color the remaining part. If we color $v_0$ red, then the second node must be green and then there must be at least one red node from $v_2$ to $v_{n-1}$ (since there must be an even number of red nodes). Let $v_j$ be the smallest $j$ between $2$ and $n-1$ for which $v_j$ is red. If $j\le n-3$, then $v_{j+1}$ must be green and the remaining part can be colored in $p(n-j-2)$ ways. If $j=n-2$, then $v_{n-1}$ must be colored green, which gives 1 coloring. $j=n-1$ also gives one coloring. This justifies the recursion  $p(n)=p(n-1)+p(n-4)+\cdots+p(1)+2$ for $n\ge 5$. This is a Fibonacci-type of sequence, which can be lower and upper bounded by $\Phi^n$, up to a constant factor. Thus, we conclude that $r(n)=\Theta(\Phi^n)$. We clearly have $r(n)\le p(n)$. Furthermore, if we color $v_0$ in  a cycle $C_n = (v_0,\cdots, v_{n-1})$ green, then the remaining nodes can be colored in $p(n-1)$ ways. Thus, we have $p(n-1)=\Theta(\Phi^{n-1})\le r(n)\le p(n)=\Theta(\Phi^n)$. This implies that $s(n)=\Theta(\Phi^n)$ since $s(n)=2r(n)$. (Actually if we solve the recursion accurately, we get $s(n)$ is equal to $\Phi^n$, up to some additive terms of smaller orders.)



\subsection{Proof of Theorem~\ref{min-winning-rmm}}
\label{min-winning-rmm-appendix}
Consider an arbitrary positive integer $k$. According to Theorem 1.1 in~\cite{berger2001dynamic}, there is an $n$-node graph $G=(V= \{v_1,\cdots,v_n\},E)$, for some $n\ge k$, where the set $\mathcal{B}:=\{v_1,\cdots, v_{18}\}$ forms a winning set in MM. Consider two copies of graph $G$, namely $G_1=(V_1=\{v_1^{(1)},\cdots, v_n^{(1)}\},E_1)$ and $G_2=(V_2=\{v_1^{(2)},\cdots, v_n^{(2)}\},E_2)$. To construct our desired graph $H$, let us add the edge $\{v_i^{(1)}, v_i^{(2)}\}$ for each $1\leq i \le n$ if $d^G(v_i)$ is even.

Consider a model M which is the same as MM but with a different tie-breaking rule. Let $\mathcal{C}_t(v)$ denote the color of node $v\in V$ at round $t\in \mathbb{N}$ in MM on $G$ assuming that $\mathcal{C}_0|_{\mathcal{B}}=b$ and $\mathcal{C}_0|_{V\setminus\mathcal{B}}=w$. Let $\mathcal{C}_t'(v)$ be the color of node $v\in V_H= V_1\cup V_2$ for $t\in \mathbb{N}$ in the model M on $H$ assuming that $\mathcal{C}'_0|_{\mathcal{B}_H}=b$ and $\mathcal{C}'_0|_{V_H\setminus \mathcal{B}_H}=w$, where $\mathcal{B}_H := \{v_1^{(1)},\cdots, v_{18}^{(1)}\}\cup \{v_1^{(2)},\cdots, v_{18}^{(2)}\}$. We claim that 
for every $t\in \mathbb{N}$ and $1 \le i \le n$, we have $\mathcal{C}_t(v_i)=\mathcal{C}^{\prime}_t(v_i^{(1)})=\mathcal{C}^{\prime}_t(v_i^{(2)})$. Combining the last statement with the fact that $\mathcal{C}_T|_V=b$ for some $T\in \mathbb{N}$ implies that $\mathcal{C}'_T|_{V_H}=b$. Thus, $H$ is a graph with more than $k$ nodes which has a winning set of size $36$ in the model M.

Using induction, we prove that for every $t\in \mathbb{N}$ and $1 \le i \le n$, we have $\mathcal{C}_t(v_i)=\mathcal{C}^{\prime}_{t}(v_i^{(1)})=\mathcal{C}^{\prime}_t(v_i^{(2)})$. This is true for the base case of $t=0$ by construction. As the induction hypothesis, assume that the statement is true for some $t-1\ge 0$. We show that it also holds for $t$. Consider an arbitrary $1\le i \le n$. If $d^G(v_i)$ is odd, then $d^G(v_i)=d^{H}(v_i^{(1)})= d^{H}(v_i^{(2)})$ and $v_i^{(1)}$ (similarly $v_i^{(2)}$) has exactly the same number of blue nodes in $\mathcal{C}'_{t-1}$ as node $v_i$ in $\mathcal{C}_{t-1}$ by the induction hypothesis. Furthermore, since the degree is odd, there is no tie-breaking, i.e., the update is the same for M and MM. Thus, we will have $\mathcal{C}_t(v_i)=\mathcal{C}^{\prime}_t(v_i^{(1)})=\mathcal{C}^{\prime}_t(v_i^{(2)})$. Now, assume that $d^G(v_i)$ is even. Let us focus on the color of $v_i^{(1)}$ in round $t$. (The same argument works for $v_i^{(2)}$.) If $|N_{b}^{\mathcal{C}_{t-1}}(v_i)|> |N_{w}^{\mathcal{C}_{t-1}}(v_i)|$, then we actually know that $|N_{b}^{\mathcal{C}_{t-1}}(v_i)|\ge |N_{w}^{\mathcal{C}_{t-1}}(v_i)|+2$ since $d^G(v_i)$ is even. Thus, by the induction hypothesis, the difference between the number of blue nodes and white nodes in the neighborhood of $v_i^{(1)}$ in $G_1$ in the $(t-1)$-th round is at least 2. This implies that $v_i^{(1)}$ chooses blue color under model M in the next round, regardless of the color of its other neighbor, namely $v_i^{(2)}$. A similar argument works for $|N_{w}^{\mathcal{C}_{t-1}}(v_i)|> |N_{b}^{\mathcal{C}_{t-1}}(v_i)|$. It remains to consider $|N_{w}^{\mathcal{C}_{t-1}}(v_i)|=|N_{b}^{\mathcal{C}_{t-1}}(v_i)|$. In this case, node $v_i$ keeps its color in round $t$. This is also true for $v_i^{(1)}$ because it has exactly the same number of blue and white neighbors in $G_1$ and thus it chooses the color of its additional neighbor, i.e., $v_i^{(2)}$, which has the same color by the induction hypothesis. Thus, it also keeps its color, regardless of the tie-breaking rule. (In general, there is no tie since all nodes in $H$ have odd degrees.)

\subsection{Proof of Theorem~\ref{cycle-winning-set}}
\label{cycle-winning-set-appendix}
For RMM, it suffices to prove that if there is even one white node in the initial coloring, there is a non-zero probability that the process does not reach the blue coloring. Let one white node form an alternating path of length one and the rest of the cycle be blue. Then, it is possible that the alternating path grows from both sides in each round. After $\lceil n/2\rceil-1$ rounds, the process reaches the blinking configuration (if $n$ is even) and a coloring with two adjacent white nodes (if $n$ is odd). In the first case the process never reaches the blue coloring and in the second one it is possible that this white path grows in each round until the process reaches the white coloring. Thus, there is no winning set of size $n-1$ or smaller.

Consider a winning set $\mathcal{B}$ in MM. For every two adjacent nodes, at least one must be in $\mathcal{B}$. This is true because otherwise if initially only nodes in $\mathcal{B}$ are blue such two adjacent nodes are colored white and remain white forever, which is in contradiction with $\mathcal{B}$ being a winning set. This implies that $|\mathcal{B}|\ge \lceil n/2\rceil$. For odd $n$, this implies that $|\mathcal{B}|\ge \lfloor n/2 \rfloor +1$. For even $n$, if there are no two adjacent nodes outside $\mathcal{B}$ and $|\mathcal{B}|= \lceil n/2\rceil=n/2$, then it means only nodes in odd (or even) position are in $\mathcal{B}$. In that case, $\mathcal{B}$ is not a winning set because starting from a coloring where only $\mathcal{B}$ is blue, the process is in the blinking configuration. Therefore, in the even case, we have $|\mathcal{B}|\ge n/2+1=\lfloor n/2\rfloor+1$. Furthermore, the bound of $\lfloor n/2\rfloor+1$ is tight. For both odd and even $n$, the set $\{v_i: (i \mod 2) =1\}\cup \{v_0\}$ is a winning set of size $\lfloor n/2\rfloor+1$.


\subsection{Proof Sketch of Corollary~\ref{corollary}}
\label{corollary-appendix}

For odd $n$, according to Theorem~\ref{cycle-periodicity-mm}, the process must reach the white or blue coloring. Let $b_f$ denote the number of blue nodes in the final coloring. We have $\mathbb{E}[b_f]=\sum_{i=1}^{n} i*\mathbb{P}[b_f=i]=n*\mathbb{P}[b_f=n]$, where for the last equality we used the above statement. Furthermore by Theorem~\ref{random-rmm}, we have $\mathbb{E}[b_f]=pn$. Therefore, we get $\mathbb{P}[b_f=n]=p$.

For even $n$, since graph $C_n$ is a bipartite graph, its node set can be partitioned into two subsets $L$ and $R$, which both form an independent set of size $n/2$. According to Theorem~\ref{cycle-periodicity-mm} after $f$ rounds, for some even integer $f$, all nodes in $L$ (similarly in $R$) share the same color. Using a similar argument to the one for the odd case and the fact that $L$ and $R$ are symmetric, we can show that the probability that all nodes in $L$ (similarly $R$) are blue in round $f$ is equal to $p$.

Furthermore, by a simple inductive argument, one can show that the color of nodes in $L$ (similarly $R$) in round $t$ for some even integer $t$ only depends on the color of nodes in $L$ (resp. $R$) in round $0$ (i.e., their initial coloring). This implies that the color of nodes in $L$ is independent of the color of nodes in $R$ in round $t$.

Combining the statements from the last two paragraphs, we can conclude that in round $f$ both $L$ and $R$ are blue (i.e., the fully blue coloring) w.p. $p^2$, both $L$ and $R$ are white (i.e., the fully white coloring) w.p. $(1-p)^2$, and one of them is blue and the other one is white (i.e., the blinking configuration) w.p. $2p(1-p)$.
\end{document}